\documentclass[preprint]{elsarticle}
	\usepackage[latin1,utf8]{inputenc} 
	\usepackage[english]{babel}
	\usepackage[T1]{fontenc}
	\usepackage{amsmath,amsopn,amssymb,amsthm}
	\newtheorem{thm}{Theorem}
	
	\newtheorem{lem}{Lemma}
	\newtheorem{rem}{Remark}

	\newtheorem{mdef}{Definition}
	\newtheorem{prop}{Proposition}

	  \let \r = \varrho \let\ph=\varphi    
	
	\def\e{{\rm e}}
	\def\i{{\rm i}}
	\def\eps{{\varepsilon}}
	\def\Id{{\rm Id}}
	\def\op{{\rm op}}
	\def\opaw{{\rm op}_{\rm aw}}
	\def\R{\mathbb R}
	\def\Rd{{\mathbb{R}^d}}
	\def\Rdd{{\mathbb{R}^{2d}}}
	
	\def\C{\mathbb C}
	\def\N{\mathbb N}

	\def\W{{\mathcal W}}

	\def\dd{{\rm{d}}}
	
	\def\nn{\nonumber}
	\def\({\left(} \def\){\right)}  
	   \def\lk{\,\left[ \,} \def\rk{\,\right] \,} 
	      
	   \def\lw{\left\langle} \def\rw{\right\rangle}

\begin{document}

\begin{frontmatter}

\author{Johannes Keller} 
\address{Zentrum Mathematik, Technische Universität München, Boltzmannstr. 3, 85748 Garching bei München}
\title{The spectrogram expansion of Wigner functions}

\begin{abstract}
Wigner functions generically attain negative values and hence are not probability densities.
We prove an asymptotic expansion of Wigner functions in terms of Hermite spectrograms, which
are probability densities. The expansion provides exact formulas for the quantum expectations of polynomial
observables. In the high frequency regime it allows to approximate quantum expectation values up
to any order of accuracy in the high frequency parameter. We present a Markov Chain Monte Carlo method
to sample from the new densities and illustrate our findings by numerical experiments.
\end{abstract}

\end{frontmatter}


\section{Introduction}

Highly oscillatory functions  $\psi \in L^2(\Rd)$, $d\geq 1$, play a prominent role in many areas of science, including
quantum molecular dynamics, wave mechanics, and quantum optics. 
The semiclassical analysis and algorithmic simulation
of such systems often requires a represention of  $\psi$ on the classical phase
space $T^*\Rd \cong\Rdd$. In this paper we construct novel phase space representations that are well-suited
for numerical sampling purposes.
 
As usual, we assume that $\psi$ is $L^2$-normalized
and oscillates with frequencies of size $O(\eps^{-1})$, where $0<\eps \ll 1$ is a small parameter.  
Then, representing $\psi$ via its Wigner transform
\begin{equation}\label{eq:wigner}
\W_\psi(q,p) = (2\pi\eps)^{-d} \int_\Rd \e^{\i p y/\eps} \psi(q-\tfrac{y}2)\overline\psi(q+\tfrac{y}2)\dd y, \quad (q,p)\in \Rdd,
\end{equation}
facilitates to express expectation values of Weyl quantized operators $\op(a)$
exactly via the weighted phase space integral
\begin{equation}\label{eq:wigner_exp_exact}
\lw \psi ,\op(a)\psi\rw = \int_\Rdd a(z) \W_\psi(z) \dd z;
\end{equation}
see, e.g., \cite[\S9 and \S10.1]{dG11}.
Despite its favorable properties, using Wigner functions has a major drawback for applications:
In chemical physics quantum expectation values are often computed via a Monte Carlo discretization
of~\eqref {eq:wigner_exp_exact}; see~\eqref{eq:wigner_method_exp} and \cite{TW04,KL14}. However,
Wigner functions generically attain negative values  and, hence, are not probability densities.
Consequently, they often cannot be sampled directly, and discretizing~\eqref{eq:wigner_exp_exact} becomes difficult or even unfeasible.

Convolving $\W_\psi$ with another Wigner function results in a so-called spectrogram, which is a
nonnegative function. For a Gaussian wave packet  $g_0$
  centered in the origin, the spectrogram $S^{g_0}_\psi :=  \W_\psi * \W_{g_0}$
is a smooth probability density known as the Husimi function of $\psi$. Since one can
sample from $S^{g_0}_\psi$, it suggests itself for replacing  the Wigner function 
in~\eqref{eq:wigner_exp_exact}.
 However, this heavily deteriorates the results by introducing errors of order $O(\eps)$,
\begin{equation}\label{eq:husimi_expec}
\lw \psi ,\op(a)\psi\rw = \int_\Rdd a(z) S^{g_0}_\psi(z) \dd z + O(\eps),
\end{equation}
see~\cite{KL13}. This is often far from being satisfactory. 

In~\cite{KLO15} we recently introduced a novel phase space density $\mu^2_\psi$, given as
a linear combination of the  Husimi function  $S^{g_0}_\psi$ and spectrograms
 associated with  first order Hermite functions. Using $\mu^2_\psi$ 
 instead of the Husimi function improves the errors in \eqref{eq:husimi_expec}  
 to order $O(\eps^2)$.

It turns out that --- as conjectured in~\cite[\S10.5]{K15} --- the results from~\cite[Theorem 3.2]{KLO15} 
can be generalized in a systematic way. We
provide a procedure to construct spectrogram
approximations with errors of arbitrary order $O(\eps^N)$, $N\in \N$. 
  Our main results are summarized in Theorem~\ref{thm:spec_exp}.
We introduce novel
phase space densities $\mu^N_\psi$ by suitably combining Hermite spectrograms of $\psi$ 
of order less than $N$. Then, using these densities 
gives the approximation
\begin{equation}\label{eq:Nth_order_approx}
\lw \psi ,\op(a)\psi\rw = \int_\Rdd a(z) \mu^N_\psi(z) \dd z + O(\eps^N), \quad N\in \N,
\end{equation}
where the error term vanishes as soon as $a$ is a polynomial of degree less than $2N$.
This approximation is well-suited for computing 
quantum expectations with high accuracy:
One only needs to sample from the densities $\mu^N_\psi$, which are linear combinations
of smooth probability densities.
We provide a Markov chain
Monte Carlo method for the sampling that merely requires quadratures
of inner products of $\psi$ with shifted Hermite functions. 

Our approximation indicates a way to circumvent the sampling problem for Wigner functions and, hence, 
might be useful in various applications.
Moreover, the spectrogram expansion 
provides insight into the structure of Wigner functions that can be employed for developing new characterizations and
approximations of functions in phase space.
An important application of our result lies in quantum molecular dynamics:
 one can approximate the quantum evolution of expectation values by sampling
from the density $\mu_\psi^N$ associated with the initial state and combine it with suitable semiclassical approximations 
for the dynamics; see~\S\ref{sec:quant_dyn} and~\cite{GL14,BR02}.


\subsection{Outline}

After recalling Wigner functions and spectrograms in \S\ref{sec:phase_repr}, 
in \S\ref{sec:spec_expansion} we present our
main results.
The proof is prepared and completed in \S\ref{sec:laplace_laguerre} and \S\ref{sec:proof_main_res}, respectively,
and~\S\ref{sec:example_densities} contains illustrative examples.

In \S\ref{sec:quant_exp} and~\S\ref{sec:MCMC} we explore the application of 
our new density for the computation of 
quantum expectations, and present a Metropolis sampling
method. In~\S\ref{sec:quant_dyn} we briefly discuss applications in quantum dynamics.

Finally, in~\S\ref{sec:accuracy} and~\S\ref{sec:hat}, we 
present numerical experiments that illustrate the validity and applicability
of our results and methods.

\subsection{Related Research}

Spectrograms and combinations of spectrograms 
have been extensively studied in the context of time-frequency analysis, 
e.g. for signal 
reassignment~\cite{AF13}, filtering~\cite{Fl15} and cross-entropy minimization~\cite{LPH94}. 
However, to the best of our knowledge, apart from our preceding work~\cite{KLO15}, there are no results
on the combination of spectrograms for approximating
Wigner functions and expectation values.

Husimi functions have been widely used in the context of quantum optics and quantum dynamics, see, e.g.,~\cite{AMP09,Schl11}~and~\cite[\S2.7]{F89}.
 In~\cite{KL13} one can find second order approximations for the quantum evolution of expectation values with
Husimi functions and corrected operator symbols.

\section{Phase space representations via spectrograms}\label{sec:phase_repr_spec}

\subsection{High frequency functions in phase space}\label{sec:phase_repr}
We start by reviewing several representations of
functions $\psi \in L^2(\Rd)$ by real-valued distributions on phase space;
see also \cite{KLO15} and \cite{dG11} for more details.

The most prominent phase space representation of $\psi$ is given by its Wigner function $\W_\psi$ defined in~\eqref{eq:wigner}. 
It has the property that expectation values of Weyl quantized operators
\begin{equation}\label{eq:weyl_quant}
(\op(a)\psi)(q) = (2\pi \eps)^{-d} \int_\Rdd a(\tfrac12(y+q),p) \e^{\i (q-y)p/\eps} \psi(y) \dd y\, \dd p
\end{equation}
with sufficiently regular symbol $a:\Rdd \to \C$
can be exactly expressed via the weighted phase space integral~\eqref{eq:wigner_exp_exact}.

Whenever $ \W_\psi$ is a probability density,~\eqref{eq:wigner_exp_exact} suggests
to approximate expectation values by means of a Monte Carlo type quadrature, see~\S\ref{sec:quant_exp}.
However, as soon as $\psi$ is not a Gaussian, $\W_\psi$ attains negative values (see~\cite{SC83,J97})
and hence is not a probability density. This imposes severe difficulties for computations, since $ \W_\psi$ cannot be sampled directly.

One can turn $\W_\psi$ into a nonnegative function by convolving it with
another Wigner function. For $\psi \in L^2(\Rd)$ and a Schwartz class window $\phi \in \mathcal{S}(\Rd)$, $\|\psi\|_{L^2} = \|\phi\|_{L^2} = 1$, the convolution
\[
S^\phi_\psi := \W_\psi * \W_\phi: \Rdd \to \R
\]
is a smooth probability density,  as can be deduced from \cite[Proposition 1.42]{F89}.
In time-frequency analysis $S^\phi_\psi$ is called a \emph{spectrogram} of $\psi$;
see, e.g., the introduction in~\cite{F13}. Spectrograms belong to Cohen's class
of phase space distributions; see \cite[\S3.2.1]{F99}.

A popular window function is provided by the Gaussian wave packet 
\begin{equation}\label{eq:gaussian_wp}
g_{(q,p)}(x) = (\pi\eps)^{-d/4} \exp\(- \tfrac{1}{2\eps}|x-q|^2 + \tfrac\i\eps p\cdot (x-\tfrac12q)\), \quad (q,p)\in \Rdd, 
\end{equation} 
centered in the origin $q=p=0$;
see~\eqref{eq:gaussian_wp}. The corresponding spectrogram
\begin{equation}\label{eq:husimi_def}
S_\psi^{g_0}(z) = \int_\Rdd \W_\psi(w) (\pi\eps)^{-d} \e^{|z-w|^2/\eps}~\dd w
\end{equation}
is known as the \emph{Husimi function} of $\psi$, first introduced in~\cite{H40}. 
By~\eqref{eq:wigner_exp_exact} one has
\begin{equation}\label{eq:husimi_exp}
\int_\Rdd a(z) S_\psi^{g_0}(z)  \dd z =\int_\Rdd (\W_{g_0}*a)(z) \W_\psi(z) \dd z =   \lw \psi, \opaw (a) \psi\rw,
\end{equation}
where $\opaw(a) = \op(\W_{g_0}*a)$ is the so-called anti-Wick quantized operator associated with $a$; see \cite[\S2.7]{F89}.

As a more general class of windows, we consider  the eigenfunctions $\{\ph_k\}_{k\in\N^d}\subset L^2(\Rd)$ of the harmonic oscillator
\[
- \tfrac{\eps^2}2\Delta_q + \tfrac12|q|^2, \quad q \in \Rd.
\]
It is well-known that $\ph_k$ is a rescaled multivariate Hermite function and,
in particular,  $\ph_0 = g_0$. The corresponding Wigner functions take the form
\begin{equation}\label{eq:Wigner-Laguerre}
\W_{\ph_k}(z) = (\pi \eps)^{-d} \e^{-|z|^2/\eps} (-1)^{|k|} \prod_{j=1}^d L_{k_j}\(\tfrac2\eps |z_j|^2\)
\end{equation}
where $z=(q,p)\in\Rdd$, $z_j = (q_j,p_j)\in\R^2$, and $L_n$ denotes the $n$th Laguerre polynomial
\begin{equation}\label{eq:laguerre_pol}
L_n(x) = \sum_{j=0}^n {n \choose n-j} \frac{(-x)^j}{j!} , \quad n\in \N, \quad x \in \R;
\end{equation}
see, e.g., \cite[\S1.9]{F89} and~\cite[\S1.3]{Th93}. The Laguerre
connection~\eqref{eq:Wigner-Laguerre} will play a crucial role
in our proof of the spectrogram expansion.

\subsection{The spectrogram expansion}\label{sec:spec_expansion}
In this section we present the core result of our paper, which is
the asymptotic expansion of Wigner functions in terms
of Hermite spectrograms. We start by taking a  closer look on the connection between Weyl and anti-Wick operators.

\begin{lem}\label{lem:weyl--aw}
Let $\eps>0$, $a:\Rdd \to \R$ be a Schwartz function and $N\in \N$. Then, there is a family of Schwartz 
functions $r_N^\eps:\Rdd \to \R$ and a constant $C>0$ independent of $a$ and $\eps$ with 
\[
\sup_{\eps>0}\|\op(r_N^\eps)\|_{L^2\to L^2}<C \sup_{|\alpha|,|\beta|\leq \lceil \tfrac d2 \rceil +1} \| \partial_q^\alpha \partial_p^\beta a^{(2N)}\|_\infty
\]
such that
\[
\op(a) = \opaw\(\sum_{j=0}^{N-1} \frac{(-\eps)^k}{4^k k!}\Delta^k a\) + \eps^{N} \op(r_N^\eps),
\]
where anti-Wick quantization has been defined in~\eqref{eq:husimi_exp}.
\end{lem}

\begin{proof}[Sketch of proof]
The assertion has been shown in~\cite[Lemma 1 and 2]{KL13}, see also~\cite[Proposition 2.4.3]{L10}. The proof builds on a Taylor expansion
of $a$ around the point $z$ in the convolution integral
\[
(\W_{g_0} * a)(z) = (\pi \eps)^{-d}\int_\Rdd a(\zeta) \e^{-|z-\zeta|^2/\eps} \dd \zeta
\]
that defines the Weyl symbol of $\opaw(a)$.
\end{proof}

We can combine Lemma~\ref{lem:weyl--aw} and~\eqref{eq:husimi_exp} in order to approximate
quantum expectation values by an integral
with respect to the Husimi function,
\[
\lw \psi,\op(a)\psi\rw = \int_\Rdd a(z) \W_\psi(z) \dd z =  \int_\Rdd \sum_{k=0}^{N-1} \frac{(-\eps)^k}{4^k k!}\Delta^k a(z) S^{g_0}_\psi(z) \dd z + O(\eps^N).
\]
Performing integration by parts on the above integral leads to the definition of a new family of smooth phase space densities.
\begin{mdef}\label{def:density}
Let $\eps>0$. For any $\psi\in L^2(\Rd)$ and $N\in \N$ we define
\[
\mu_\psi^{N}:\Rdd \to \R , \quad \mu_\psi^{N}(z) := \sum_{k=0}^{N-1}  \frac{(-\eps)^k}{4^k k!}\Delta^k S^{g_0}_\psi(z),
\]
where  $S^{g_0}_\psi = \W_\psi * \W_{g_0}$ is the Husimi transform of $\psi$.
\end{mdef}

Our following main Theorem shows that $ \mu_\psi^{N}$ can be used to 
replace the Wigner function $\W_\psi$ for approximating expectation values of Weyl quantized operators with $O(\eps^N)$ accuracy. 
Moreover, $ \mu_\psi^{N}$ can be written as a linear combination of Hermite spectrograms. 

\begin{thm}[Spectrogram expansion]\label{thm:spec_exp}
Let $\psi \in L^2(\Rd)$, $N\in \N$, and $\eps>0$. Then, the phase space function $\mu_\psi^{N}$
can be expressed in terms of Hermite spectrograms,
\begin{equation}\label{eq:def_mudens}
\mu_\psi^{N}(z) = \sum_{j=0}^{N-1} (-1)^j C_{N-1,j} \sum_{\substack{k\in \N^d \\ |k|=j}} S_\psi^{\ph_k}(z), \quad
C_{k,j} = \sum_{m=j}^k 2^{-m} {d-1+m \choose d-1 + j};
\end{equation}
see also Definition~\ref{def:density}.
Furthermore, if $a:\Rdd \to \C$ is a Schwartz function, there is a constant $C\geq 0$ such that
\begin{equation}\label{eq:spec_approx}
\bigg| \int a(z) \W_\psi(z)\dd z- \int_\Rdd a(z) \mu_\psi^{N}(z) \dd z \bigg| \leq C \eps^{N} \|\psi\|^2_{L^2} ,
\end{equation}
where $C$ only depends on bounds on  derivatives of $a$ of degree $2N$ and higher.
In particular, if $a$ is a polynomial of maximal degree $\deg(a)<2N$, one has $C=0$ and
the error in~\eqref{eq:spec_approx} thus vanishes.
\end{thm}

We postpone the proof of Theorem~\ref{thm:spec_exp} to chapter~\S\ref{sec:proof_main_res}. 
Firstly, in~\S\ref{sec:laplace_laguerre},
we derive an expansion for iterated Laplacians of~$\W_{g_0}$. This is
the main ingredient for identifying
$\mu_\psi^N$ with a linear combination of Hermite spectrograms.

The second order version of Theorem~\ref{thm:spec_exp} has already been shown
in~\cite[Theorem 3.2 and Proposition 3.4]{KLO15}. There, we proved that one has
\[
\mu_\psi^{2}(z) = (1 + \tfrac{d}2) S_\psi^{g_0} - \tfrac12 \sum_{j=1}^d S_\psi^{\ph_{e_j}}
\]
as well as
\begin{equation}\label{eq:spec_approx2}
\bigg| \int a(z) \W_\psi(z)\dd z- \int_\Rdd a(z) \mu_\psi^{2}(z) \dd z \bigg| \leq C \eps^{2} \|\psi\|^2_{L^2}.
\end{equation}
for a constant $C>0$ depending on third and higher derivatives of $a$.

\begin{rem}
Theorem~\ref{thm:spec_exp} remains true for more general operators $\op(a)$ 
as long as $a$ is sufficiently regular; see also~\cite[\S4.4]{Z12}. If 
$\op(a)$ is unbounded, one has to choose $\psi$ from a suitable subset of $L^2(\Rd)$.
\end{rem}

\begin{rem}\label{eq:weak_approx}
The approximation~\eqref{eq:spec_approx} of
expectation values can also  be seen  as a weak approximation of
Wigner functions. In other words, we have
\[
\W_\psi =  \mu_\psi^{N} + O(\eps^N), \quad N\in \N,
\]
in the distributional sense. This observation  is particularly interesting since
$\W_\psi$ is only continuous in general, whereas $\mu_\psi^{N}$ is always real analytic.
\end{rem}

\subsection{Iterated Laplacians of phase space Gaussians}\label{sec:laplace_laguerre}
There are many famous interrelations between the derivatives of Gaussians 
and Hermite and Laguerre polynomials; see, e.g.,~\cite{Th93} and~\cite[\S V]{Sz75}.
We present an expansion of iterated Laplacians of the 
phase space Gaussian $\W_{g_0}$ based on Laguerre polynomials. 
To the best of our knowledge, this formula did not appear in the literature before.

We aim to express the polynomial factors arising in 
iterated Laplacians of $\W_{g_0}$ as linear combinations of the product polynomials
\begin{equation}\label{eq:sum_lagu}
\mathcal L_k(\r(z)) : =\prod_{j=1}^d L_{k_j}(\r_j(z)),\quad z \in \Rdd,  \quad k\in \N^d,
\end{equation}
where we use the variables
\begin{equation}\label{eq:rho_variables}
\r_j(q,p) = \tfrac{2}\eps (q_j^2 + p_j^2), \quad j=1,\hdots, d,
\end{equation}
for readability. As known from~\eqref{eq:Wigner-Laguerre}, these polynomials  
also appear in the Wigner functions of Hermite functions.
We split our proof into two parts and treat the one-dimensional case first.

\begin{prop}\label{lem:1d_gauss_lagu} 
Let $d=1$ and $\eps>0$. Then, for all $N\in \N$ we have
\begin{equation*}
\( -\tfrac{\eps}2 \Delta \)^N \W_{g_0}(z)= N! \; \W_{g_0}(z) \sum_{n=0}^N {N \choose n} L_n(\r(z)), \quad  z\in\R^2,
\end{equation*}
where  $L_n$ is the $n$th Laguerre polynomial, and $\r$ has been defined in~\eqref{eq:rho_variables}.
\end{prop}

An induction proof of Proposition~\ref{lem:1d_gauss_lagu} can be
 found in~\ref{app:proof_lem}.

In higher dimensions one has to sum over the Laguerre products $\mathcal L_k(\r)$ 
instead of the polynomials $L_n(\r)$.
However, by applying Proposition~\ref{lem:1d_gauss_lagu}, the proof
for the multi-dimensional formula reduces to a bookkeeping exercise.

In the proof of the following Theorem we repeatedly use the binomial identity
\begin{equation}\label{eq:sum_binom_product}
\sum_{j=0}^{N-m}{N-j \choose m} {k+j \choose j}  = {N+k+1 \choose N-m},
\quad  k,N,m \in \N, \quad m\leq N.
\end{equation}
For the reader's convenience we include a short proof
of~\eqref{eq:sum_binom_product} in~\ref{sec:proof_binom}.

\begin{thm}\label{thm:gauss-laguerre}
Let $\eps>0$, $d\in \N$ and $N\in \N$. Then,
\begin{equation}\label{eq:multidi_laplace_gauss_formula}
\( -\tfrac{\eps}2 \Delta \)^N \W_{g_0}(z)= N!  \, \W_{g_0}(z)  \sum_{n=0}^N {N +d -1 \choose n+d-1} \sum_{k\in \N^d, |k|=n} \mathcal L_k(\r(z)),
\end{equation}
where $z\in\Rdd$ and the polynomials $\mathcal L_k\circ \r$ have been defined in~\eqref{eq:sum_lagu}.
\end{thm}

\begin{proof}
Since $\W_{g_0}$ is a tensor product of $d$ bivariate Gaussians of the form
\[
G(x,\xi) = (\pi \eps)^{-1}\e^{-(x^2+\xi^2)/\eps}, \quad (x,\xi)\in\R^2,
\]
the multinomial theorem implies
\begin{align}
\( -\tfrac{\eps}2 \Delta \)^{N} \W_{g_0}(z)&\nonumber= \( -\tfrac{\eps}2 (\Delta_{z_1} + \hdots + \Delta_{z_d}) \)^{N} \prod_{j=1}^d G(z_j)\\
&\nonumber= \sum_{k\in \N^d, |k|=N} {N \choose k_1,\hdots,k_d} (-\tfrac{\eps}2 \Delta)^k \W_{g_0}(z)
\end{align}
where $\Delta_{z_j} = \partial_{q_j}^2 + \partial_{p_j}^2$ and $\Delta ^k = \Delta_{z_{1}}^{k_1}\cdots \Delta_{z_d}^{k_d}$.
Consequently, after applying Proposition~\ref{lem:1d_gauss_lagu} and reordering the sum, we arrive at
\begin{align}
\( -\tfrac{\eps}2 \Delta \)^{N} \W_{g_0}(z)&\nonumber= \sum_{k\in \N^d, |k|=N} {N \choose k_1,\hdots,k_d} k! \prod_{j=1}^d \sum_{m=0}^{k_j} {k_j \choose m} G(z_j)L_{m}(\r_j(z))\\
&\label{eq:mehrdi_vor_reorder}= N!\W_{g_0}(z)  \sum_{k \in \N^d, |k|=N}  \prod_{j=1}^d \sum_{m=0}^{k_j} 
{k_j \choose m} L_{m}(\r_j(z)).
\end{align}
Now, we collect all binomial coefficients  belonging to one
polynomial $\mathcal L_\ell(\r)$ with~$0\leq |\ell| \leq N$. We treat the simple cases $|\ell|\leq 1$ separately 
in order to illustrate our counting procedure.

\begin{description}
\item[$\ell=0:$] In the sum~\eqref{eq:mehrdi_vor_reorder}, the polynomial $\mathcal L_0 \circ \r$ appears
\[
|  \{k \in \N^d:|k |=N\} |= {N+d-1\choose d-1}
\]
times. For all $k\in \N^d$ and $1\leq j \leq d$ we get the prefactor ${k_j \choose 0}  =1$.

\item[$|\ell|=1:$] For $\ell=e_i$, $i\in\{1,\hdots,d\}$, the coefficient of $\mathcal L_{e_i}\circ \r$ can be computed as follows. 
If  $k_i= N$ in~\eqref{eq:mehrdi_vor_reorder}, the binomial prefactor
is $ {N \choose 1}$. If $k_i = N-1$, there are ${d-1 \choose 1}$
ways to distribute the excessive index point, and this choice does not influence the prefactor
 $ {N-1\choose 1}$. For $k_i = N-2$ there are
${d \choose 2}$ ways to distribute the two excessive index points, and the prefactor
is ${N-2\choose 1}$. Continuing in the same way, and computing the sum via~\eqref{eq:sum_binom_product}, we obtain
\[
\sum_{j=0}^{N-1}  {N-j \choose 1}{d-2+j \choose j} = {N+d-1 \choose d},
\]
which is the coefficient of the $n=1$ term in~\eqref{eq:multidi_laplace_gauss_formula}.
\item[$|\ell|=n\leq N:$]   Without loss of generality, assume that $\ell$ has $1\leq r \leq d$ nonzero entries $\ell_1,\hdots,\ell_r >0$ and 
$\ell_{r+1},\hdots,\ell_d = 0$. Otherwise rename the coordinates. For every $k\in \N^d$
and $s \leq d $ we define the partial sums $|k|_s = k_1 + \hdots +k_s$ such that $|\ell|_r = n$.

Then, if $|k|_r = N$ in~\eqref{eq:mehrdi_vor_reorder}, one has to sum all prefactors of the form
\[
 \prod_{j=1}^r {k_j \choose \ell_j}, \quad k_j \geq \ell_j, \quad  \sum_{j=1}^r k_j=N.
\]
If $|k|_r = N-1$,  one additionally has ${d-r \choose 1}$ ways to distribute the excessive index point et cetera.
In total, all prefactors of $\mathcal L_\ell(\r(z))$ are given by
\begin{align*}
\sum_{m_1=0}^{N-n}&\sum_{m_2=\ell_{2}}^{M+|\ell|_2 - |m|_1}  \sum_{m_3=\ell_{3}}^{M+|\ell|_3-|m|_2} 
\cdots  \sum_{m_r=\ell_{r}}^{M+|\ell|_r - |m|_{r-1}} 
{N-|m|_r \choose \ell_{1}} \\
& \times {m_2 \choose \ell_{2}}  \cdots {m_r \choose \ell_{r}} 
{d-r  -1+m_1 \choose m_1} 
\end{align*}
where $m=(m_1,\hdots,m_r)\in \N^r$ and $M=N-n-\ell_{1}$. The summation over  $m_1$ captures 
all index points of $k$ in the components $r+1,\hdots,d$.
For the innermost sum we compute
\begin{align*}
&\sum_{m_r=\ell_r}^{M+|\ell|_r - |m|_{r-1}}
  {N-|m|_r \choose \ell_1} {m_r \choose  \ell_r} \\
  &= \sum_{m_r=0 }^{M+|\ell|_{r-1} - |m|_{r-1}}
  {N-|m|_{r-1}-m_r -\ell_{r}  \choose \ell_1 } {m_r + \ell_r \choose m_r} \\
  &={ N-|m|_{r-1} + 1 \choose 1+\ell_r +\ell_1 }
\end{align*}
by invoking~\eqref{eq:sum_binom_product}. Repeating this computation in a similar way for the sums over $m_{r-1} ,\hdots, m_2$
one is left with the last sum over $m_1$, which gives
\begin{align*}
\sum_{m_1=0}^{N-n} {N- m_1 + r -1 \choose n + r -1} {d-r  -1+m_1 \choose m_1}  &=  {N+d -1 \choose n+d -1},
\end{align*}
again by using~\eqref{eq:sum_binom_product}. 
\end{description}
Rewriting~\eqref{eq:mehrdi_vor_reorder} by incorporating the above calculations completes the proof.
\end{proof}

\subsection{Proof of the main result}\label{sec:proof_main_res}

We can now prove our main result Theorem~\ref{thm:spec_exp}. The central idea is to 
employ the Laplace-Laguerre formula from Theorem~\ref{thm:gauss-laguerre}, and to
identify the Laguerre polynomials with the prefactors appearing in
the Wigner functions~\eqref{eq:Wigner-Laguerre}. 
These Wigner functions in turn are the convolution kernels of Hermite spectrograms.

\begin{proof}[Proof of Theorem~\ref{thm:spec_exp}]
Let $a:\Rdd \to \R$ be an $\eps$-independent Schwartz function. Then,
by invoking \eqref{eq:husimi_exp} and Lemma~\ref{lem:weyl--aw}, we have
\begin{align*}
\lw\psi,\op(a) \psi \rw &= \int_\Rdd  \sum_{m=0}^{N-1}\frac{(-\eps\Delta)^m}{4^k m!}  a(z) ( \W_{g_0}* \W_\psi)(z) \dd z + \eps^{N} \lw \psi,\op(r_{N}^\eps)  \psi \rw
\end{align*}
where $\{r_{N}^\eps\}_{\eps>0}$ is a family  of Schwartz functions with uniformly bounded operator norm.
 Note that $r_N^\eps$ only depends on $2N$th and higher order derivatives of $a$; see also~\cite[\S2.3]{KLO15}.
Repeated integration by parts yields
\begin{align*}
\lw\psi,\op(a) \psi \rw &= \int_\Rdd  a(z)   \sum_{m=0}^{N-1}\frac{(-\eps\Delta)^m}{4^k m!} 
( \W_{g_0}* \W_\psi)(z) \dd z + \eps^{N} \lw \psi,\op(r_{N}^\eps)  \psi \rw,
\end{align*}
and we recognize the phase space density $\mu^{N}_\psi$ from Definition~\ref{def:density}.
Now, by Theorem~\ref{thm:gauss-laguerre}  we have
\begin{align*}
\frac{(-\eps\Delta)^m}{4^k m!} &
( \W_{g_0}* \W_\psi)(z)= \frac{1}{m!} 
\((-\tfrac\eps4\Delta)^m \W_{g_0}* \W_\psi\)(z) \\
&=   \(2^{-m} \W_{g_0}  \sum_{j=0}^m {m +d -1 \choose j+d-1} \sum_{k\in \N^d, |k|=j}
 \mathcal L_k(\varrho) *\W_\psi  \)  (z),
\end{align*}
and using formula~\eqref{eq:Wigner-Laguerre} leads us to
\begin{align*}
\frac{(-\eps\Delta)^m}{4^k m!} &
\(\W_{g_0}*\W_\psi\)(z)=   2^{-m} \sum_{j=0}^m (-1)^{j}  {m +d -1 \choose j+d-1} \sum_{k\in \N^d, |k|=j}
 S_\psi^{\varphi_k} (z).
\end{align*}
Finally, summing over all $m=1\hdots N-1$ and reordering the sum gives
\begin{align*}
 \sum_{m=0}^{N-1}
  2^{-m} \sum_{j=0}^m (-1)^{j}  {m +d -1 \choose j+d-1} \sum_{k\in \N^d, |k|=j}
 S_\psi^{\varphi_k} (z)
 &= 
  \sum_{j=0}^{N-1} C_{N-1,j}    \sum_{k\in \N^d, |k|=j}
 S_\psi^{\varphi_k} (z)  
\end{align*}
with
\[
C_{k,j}   = \sum_{m=j}^{k} 2^{-m} (-1)^{j}  {m +d -1 \choose j+d-1}, \quad j=0,\hdots,k,
\]
and the assertion follows.
\end{proof}

\subsection{Examples}\label{sec:example_densities}

From~\cite[Proposition 5]{LT14} we know that the Husimi functions
of the Hermite functions $\{\ph_k\}_{k\in\N^d}$ are given by the formula
\begin{equation*}\label{eq:spec_herm_hus}
S^{g_0}_{\ph_k}(z) = S^{\ph_k}_{g_0}(z) = (2\pi \eps)^{-d} \frac{\e^{-|z|^2/2\eps}}{(2 \eps)^{|k|} k!} |z|^{2k}. 
\end{equation*}
By using the covariance property of Wigner functions
with respect to Heisenberg-Weyl operators $T_z$,
\begin{equation}\label{eq:heisenberg_weyl}
T_{(q,p)}\psi = \e^{\i p(\bullet - q/2)/\eps} \psi(\bullet -q), \quad \psi \in L^2(\Rd),
\end{equation}
 see~\cite[Proposition 174]{dG11},
 one then can easily compute the 
new phase space densities $\mu_\psi^N$  for a one-dimensional Gaussian wave packet $\psi = g_w$, $w\in \R^2$.
Namely, we find the weak approximations
\begin{align}
\W_{g_w}(z)  &= \sum_{j=0}^N (-1)^j \sum_{m=j}^N 2^{-m}  {m \choose j} 
S^{\varphi_j}_{g_w}(z) + O(\eps^{N+1})\\
&\nonumber=   \sum_{j=0}^N (-1)^j 
\frac{(2\pi\eps)^{-1}}{j!}
\big| \tfrac{z-w}{\sqrt{2\eps}} \big|^{2j} \e^{-|z-w|^2/2\eps}  \sum_{m=j}^N 2^{-m}  {m \choose j} + O(\eps^{N+1})
\end{align}
for $z\in \R^2$, such that the first three nontrivial approximations read
\begin{align}
\W_{g_w}(z)&\nonumber= (2\pi\eps)^{-1}\e^{-|z-w|^2/2\eps} \(   \tfrac32  - \tfrac12 \tfrac{|z-w|^2}{2\eps}\)+O(\eps^2), \\
\W_{g_w}(z)&\label{eq:mu_gauss1d}= (2\pi\eps)^{-1}\e^{-|z-w|^2/2\eps} \(   \tfrac74 - \tfrac{|z-w|^2}{2\eps} +  \tfrac14 \tfrac{|z-w|^4}{ 2! 4\eps^2} \)+O(\eps^3), \\
\W_{g_w}(z)&\nonumber=(2\pi\eps)^{-1} \e^{-|z-w|^2/2\eps}\(   \tfrac{15}8  -  \tfrac{11}{8}   \tfrac{|z-w|^2}{2\eps} +  
\tfrac{5}{8} \tfrac{|z-w|^4}{2! 4\eps^2}  - \tfrac{1}8\tfrac{|z-w|^6}{3!8\eps^3}  \)+O(\eps^4).
\end{align}
It is striking 
 that the sequence of densities $\mu_{g_w}^N$ does not only approximate $\W_{g_w}$ weakly as $\eps \to 0$, but
even seems to yield a strong approximation as $N\to \infty$,
see Figure~\ref{fig:gauss_densities}.

\begin{figure}[h!]
\includegraphics[width=\textwidth]{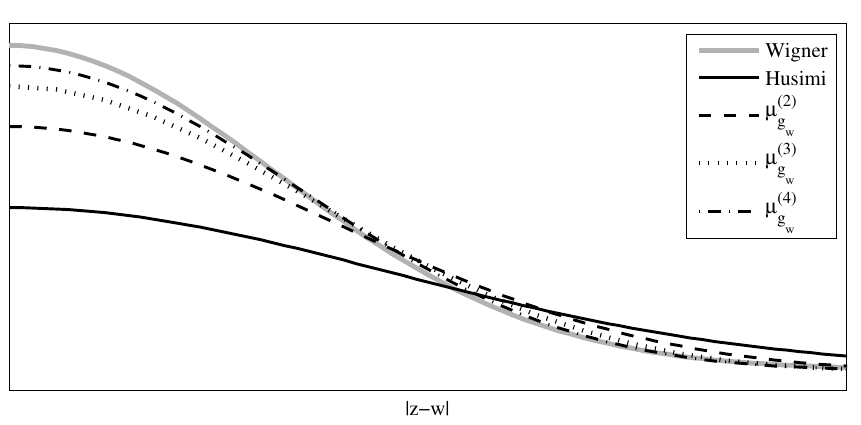}
\caption{\label{fig:gauss_densities} Decay of the Wigner and Husimi function 
as well as  the densities $\mu_{g_w}^{N}$, $N\leq 4$, 
for a Gaussian wave packet $g_w$ in dependence of the distance $|z-w|$ 
from the phase space center.}
\end{figure}

In higher dimensions
 one has to incorporate different prefactors and sum over all Hermite spectrograms
 of the same total degree, but the structure of the approximations~\eqref{eq:mu_gauss1d}
 remains the same.

For Gaussian superpositions $\psi = g_{z_1} + g_{z_2}$ with phase
space centers $z_1,z_2 \in \Rdd$,
and Hermite functions,
$\psi = \ph_k$, the second order density $\mu^2_\psi$ has been computed
in~\cite[\S5]{KLO15} by using ladder operators. The same 
 technique can in principle 
also be used to compute higher order densities $\mu^N_\psi$,
but will lead to tedious calculations. The structure of the densities
for a Gaussian superposition, however, is always of the form
\[
\mu^N_{g_{z_1} + g_{z_2}} = \mu^N_{g_{z_1}} + 
\mu^N_{g_{z_2}} + \e^{-|z_1-z_2|/8\eps} C^N_{z_1,z_2},
\]
where $C^N_{z_1,z_2}$ is an oscillatory cross term, see also~\cite[\S5]{KLO15}.
The damping factor $\e^{-|z_1-z_2|/8\eps}$ is exponentially small in $\eps$,
such that one can safely ignore the cross term in computations as soon as
 $z_1$ and $z_2$ are sufficiently apart.
In contrast, the cross term in the
Wigner function of a Gaussian superposition does not contain a damping factor.
Hence, the interferences are large and cannot be neglected. 

In this paper we do not further investigate explicit formulas for spectrogram densities.
Instead, we discuss a Markov Chain method for sampling from spectrogram densities  that is
tailored to practical applications.  In particular, the method can be applied to a wide 
range of states and circumvents the difficulties of explicitely computing
Wigner or Husimi functions, see~\S\ref{sec:MCMC}.

\section{Applications}\label{sec:applications}

\subsection{Quantum Expectations}\label{sec:quant_exp}
In chemistry, the expectation values of Weyl quantized observables
are often computed via the Monte Carlo quadrature
\begin{equation}\label{eq:wigner_method_exp}
\lw \psi, \op(a) \psi \rw =\int_\Rdd a(z)\W_\psi(z) \dd z \approx \frac1n \sum_{m=1}^n a(z_m)
\end{equation}
where $z_1,\hdots,z_n\sim \W_\psi$ are distributed with respect 
to the Wigner function, see~\cite{TW04,KL14}. Generically, however, $\W_\psi$ is not a probability density
and direct sampling techniques cannot be applied. Instead of using methods like importance sampling
we propose to replace $\W_\psi$ by a spectrogram density $\mu^N_\psi$, which
is a linear combination of smooth probability densities. That is,
we approximate
\begin{align}
\lw \psi, \op(a) \psi \rw &\nonumber=\int_\Rdd a(z)\mu^N_\psi(z) \dd z +O(\eps^N) \\
&\label{eq:approx_exp_mu}\approx   \sum_{j=0}^{N-1} (-1)^jC_{N-1,j}{N+d-1\choose d-1} \frac1n \sum_{m=1}^n a(z^j_m),
\end{align}
where the phase space points are sampled from the probability densities given by the averaged Hermite spectrograms 
of a given order,
\begin{equation}\label{eq:smapling_densities}
z^j_1,\hdots,z^j_n\sim {N+d-1\choose d-1}^{-1} \sum_{k\in \N^d,|k|=j} 
S_\psi^{\ph_k}, \quad j=1,\hdots N-1.
\end{equation}
Obviously,  method~\eqref{eq:approx_exp_mu}  is typically only practicable if the
dimension $d$ is not too large and one does not need to go to a very high order $N$. 
However, for the majority of applications in physical
chemistry  this is the case. 

\begin{rem}
Instead of considering the probability densities~\eqref{eq:smapling_densities}
it could often be more practicable to sample from each spectrogram $S_\psi^{\ph_k}$, $|k|<N$, seperately.
Alternatively, sometimes it might be possible to combine all spectrograms that appear with positive or 
negative prefactors. In that case, one would only need to sample from two probability densities.
\end{rem}

\subsection{Sampling via Metropolis-Hastings}\label{sec:MCMC}

Evaluating the highly oscillatory integral~\eqref{eq:wigner} 
defining the Wigner function in several dimensions is numerically extremely
challenging or ---  for the majority of systems --- simply unfeasible. Together with the sampling problem arising from the fact that Wigner functions may attain negative values, this is a major bottleneck for  the applicability of~\eqref{eq:wigner_method_exp}.
Moreover, often one also cannot explicitely compute the spectrogram densities~\eqref{eq:smapling_densities} either.
Instead, we propose a Markov chain sampling scheme for spectrograms based on
the inner product representation with Hermite functions
\begin{equation}\label{eq:hermite_spec_quad}
S^{\ph_k}_\psi(z) = (2 \pi \eps)^{-d} |\lw \psi, T_z \ph_k \rw|^2, \quad z \in \Rdd,
\end{equation}
where the Heisenberg-Weyl operator $T_z$ has been defined in~\eqref{eq:heisenberg_weyl};
see also~\cite{KLO15}.
This method does not require to determine $S^{\ph_k}_\psi$  globally as a function, but only
involves pointwise evaluations.
   
For approximating the inner products~\eqref{eq:hermite_spec_quad}
one can use different methods. Natural choices certainly include Gauss-Hermite,
Monte Carlo or Quasi-Monte Carlo quadrature rules. All these schemes
exploit the Gaussian factor appearing in the Hermite functions.
Monte Carlo quadrature is especially useful in higher dimensions, where one would need to 
employ sparse grids when applying Gauss-Hermite quadrature, see, e.g., \cite[\S III.1]{Lu08}. 

We propose to generate a Markov chain with stationary distribution 
$S^{\ph_k}_\psi$ via the Metropolis-Hastings algorithm. We implement the following iteration that
starts from a seed $z_0\in \Rdd$ with probability $S^{\ph_k}_\psi(z_0)$.
\begin{enumerate}
\item Proposition: set $z = z_n + \sqrt{\eps} \zeta$ with a random vector $\zeta \sim N(0,\Id_{2d})$.
\item Quadrature: approximately evaluate $S^{\ph_k}_\psi(z)$ via~\eqref{eq:hermite_spec_quad}.
\item Acceptance: generate a uniform random number $\rho \sim U([0,1])$. Accept the trial point if
$\rho < S^{\ph_k}_\psi(z)/S^{\ph_k}_\psi(z_{n})$, and set  $z_{n+1}=z$. Otherwise, reject the proposition and keep the old point $z_{n+1} = z_{n}$.
\end{enumerate}
We used a normal density of variance $\eps$ as proposal
distribution, since --- as the Husimi function of a Gaussian wave packet ---
it is a prototype spectrogram for functions with $O(\eps^{-1})$ frequencies. If one knows in 
 advance that the spectrogram $S^{\ph_k}_\psi$
has a disconnected effective support in phase space, one may additionally incorporate a jump step 
in the spirit of~\cite[\S5.1]{KLW09}.

If the Markov chain $\{z_n\}_n$ is uniformly ergodic, the central limit theorem implies weak convergence of averages, see~\cite{T94}. 
More precisely, for any function $a:\Rdd \to \R$ that is 
square-integrable with respect to  $S^{\ph_k}_\psi$ there is a constant
 $c_a$ such that
\[
\lim_{n\to \infty} \mathbb{P}\(\bigg| \tfrac1n \sum_{j=1}^n a(z_j) - \int_\Rdd a(z) S^{\ph_k}_\psi(z) \dd z \bigg| \leq \frac{r c_a}{\sqrt{n} }\) 
= \frac{1}{\sqrt{2\pi}} \int_{-r}^r \e^{-t^2/2}\dd t
\]
for any $r>0$. In particular, this implies convergence of the method~\eqref{eq:approx_exp_mu} for the computation of quantum
expectation values.
We stress that the convergence rate of $n^{-1/2}$ does not depend on the dimension $d$ of the configuration space.

\subsection{Quantum dynamics}\label{sec:quant_dyn}

In physical chemistry, the computation of stationary quantum expectation  values itself is not
of central interest. Instead, one would like to compute the evolution of expectation values
\[
t \mapsto \lw \psi_t,\op(a) \psi_t\rw.
\]
where the wave function $\psi_t$ represents the state of the molecule's nuclei at time $t$
in the Born-Oppenheimer approximation.
Here, the evolution of the wave function $\psi_t$ on an electronic potential energy surface $V$ is typically
described by the bona fide Schr\"odinger equation
\[
i \eps \partial_t \psi_t = -\tfrac{\eps^2}2\Delta \psi_t + V \psi_t,
\]
where the small parameter $0<\eps\ll 1$ represents 
the square root of the electronic versus average nuclear mass; see~\cite{ST01}. 
Consequently, by combining Egorov's theorem (see~\cite{BR02,KL14}) with~\eqref{eq:spec_approx2} 
one obtains the second order approximation
\begin{align}\label{eq:egorov_wig}
\lw \psi_t,\op(a) \psi_t\rw &= \int_\Rdd \W_{\psi_0}(z)(a\circ \Phi^t)(z)\dd z + O(\eps^2)  \\
&\label{eq:egorov_mu2}= \int_\Rdd \mu^2_{\psi_0}(z)a(z)\dd z + O(\eps^2),
\end{align}
 where $\Phi^t$ is the
the flow of the underlying classical Hamiltonian system $\dot q = p$, $\dot p = -\nabla V(q)$. 
This approximation and its discretization has  been studied in~\cite{KLO15}. 

The spectrogram method~\eqref{eq:egorov_mu2}  improves the Wigner function method~\eqref{eq:egorov_wig} that 
has been widely used in chemical physics since decades
under the name linearized semiclassical initial value representation (LSC-IVR) or 
Wigner phase space method; see, e.g., \cite{TW04,KL14}.

One can construct higher order versions of~\eqref{eq:egorov_mu2} that only require 
sampling from probability
densities and solving ordinary equations. For this purpose one combines the densities $\mu^N_{\psi_0}$
from Theorem~\ref{thm:spec_exp} with higher order corrections
of Egorov's theorem for the quantum dynamics, see~\cite{BR02,GL14}.
We leave the details to future investigations.

\section{Numerical Experiments}

\subsection{Accuracy}\label{sec:accuracy}

In a first set of experiments we investigate if the asymptotic error of our approximation from Theorem~\ref{thm:spec_exp}
is observed in practice. For this purpose
we consider a one-dimensional Gaussian wave packet $\psi = g_{z_0}$ centered in $z_0=(\tfrac12,-1)$ and varying values
of $\eps$. We compare the
expectation values of the following observables with their approximation via the spectrogram approximation with 
density $\mu_\psi^N$ for the orders $N=1,\hdots,4$:
\begin{enumerate}
\item $a(q,p) = q^4+1$
\item $b(q,p) = \tfrac14(p^2 - q )^3$
\item $c(q,p) = \cos(q)$
\item  $d(q,p) = \exp(\sin(q))$.
\end{enumerate}
We used the formulas for the spectrogram densities $\mu_\psi^N$ from~\eqref{eq:mu_gauss1d}. For the observables
$a$, $b$ and $c$ all computations can be done explicitely. For the observable $d$ we used a highly accurate quadrature scheme.
The results depicted in figure~\ref{fig:gauss_error}  show that the errors are indeed of order $O(\eps^{N})$. Moreover, as expected,
in the cases $N=3$ and $N=4$ the observables $a$ respectively $a$ and $b$ are reproduced without error.

\begin{figure}[h!]
\includegraphics[width=\textwidth]{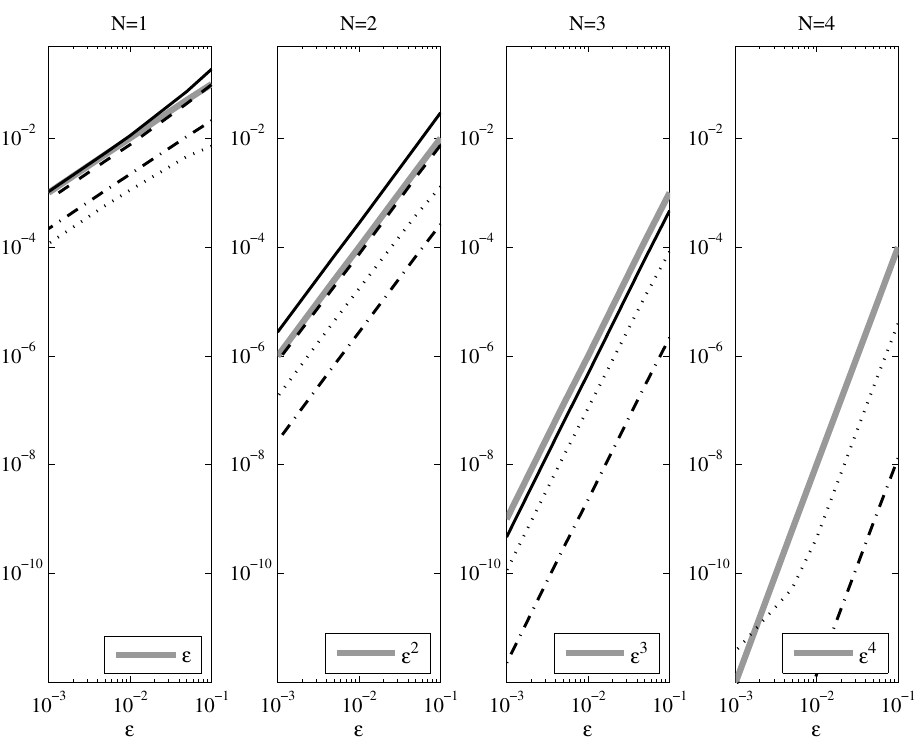}
\caption{\label{fig:gauss_error} Errors of the expectation values of the observables $a$ (dashed), $b$ (solid), $c$ (dashed dotted) and $d$ (dotted)
computed by the  spectrogram approximations of order $N\in\{1\hdots 4\}$. The state is a 
Gaussian wave packet centered in $(\tfrac12,-1)$.
}
\end{figure}

We highlight that the error constants do not seem to grow with the order,
although $\mu^N_\psi$ only weakly approximates $\W_\psi$. This indicates that 
stronger types of convergence might hold for particular states and observables.

\subsection{Sampling a hat function}\label{sec:hat}

In a second set of experiments we consider the normalized semiclassical hat function
\begin{equation}\label{eq:hat}
\psi(x) = \sqrt{\tfrac3{2\sqrt{\eps}}} \(1-\frac{|x-q|}{\sqrt\eps}\) 1_{|x-q|<\sqrt{\eps}}, \quad x\in\R,
\end{equation}
that is localized around $q \in \R$. Computing the Wigner functions and spectrograms of $\psi$
explicitely is difficult. Therefore, we sample from the densities $\mu_\psi^N$ by means of the 
Markov chain Monte Carlo algorithm introduced in \S\ref{sec:MCMC},
and discretize the inner prduct~\eqref{eq:hermite_spec_quad} by Quasi-Monte Carlo
quadrature with $10^3$ Sobol points. 

In figure~\ref{fig:hat}  one can see that the numerically computed Wigner function and its approximative reconstruction via the weighted histogram
\begin{equation}\label{eq:weighted_hist}
z \mapsto \frac1n \sum_{j=0}^2 (-1)^jC_{2,j}\#\{k: z^j_k \approx z\}, \quad z^j_1,\hdots, z^j_n \sim S_\psi^{\ph_j}
\end{equation}
of the signed density $\mu^3_\psi$ look very similar. In fact, the weighted histogram attains negative
values in the same regions where also the Wigner functions becomes negative.

\begin{figure}[h!]
\includegraphics[width=0.48\textwidth]{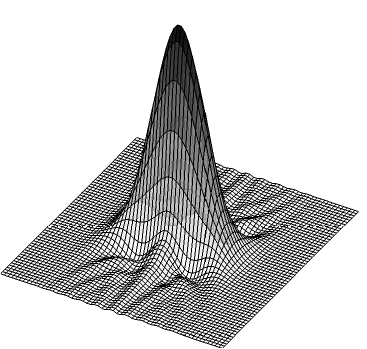}
\includegraphics[width=0.48\textwidth, height=0.4\textwidth]{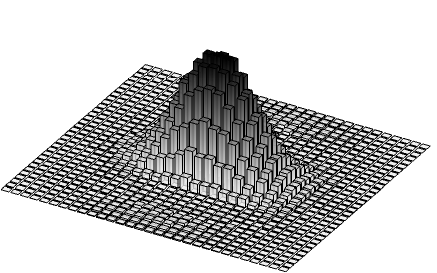}
\caption{\label{fig:hat} The Wigner function of the hat function~\eqref{eq:hat}
and a weighted histogram reconstruction~\eqref{eq:weighted_hist} of the spectrogram 
density~$\mu^3_\psi$ with $10^6$ samples, where $\eps = 5\cdot 10^{-2}$.
}
\end{figure}

In order to investigate the applicability of the Markov chain sampling algorithm from \S\ref{sec:MCMC}
we now explore the errors for observables in dependence
of the chosen number of Monte Carlo points.  
We consider the expectations of the position observable $a(q,p)=q$, that is, the center of the sampled distribution,
as well as of the complicated observable $d$ from the previous section. We consider samplings of both a single Husimi function and
the second order spectrogram density $\mu_\psi^{2}$ for the fixed parameter $\eps = 10^{-2}$.

\begin{figure}[h!]
\includegraphics[width=\textwidth]{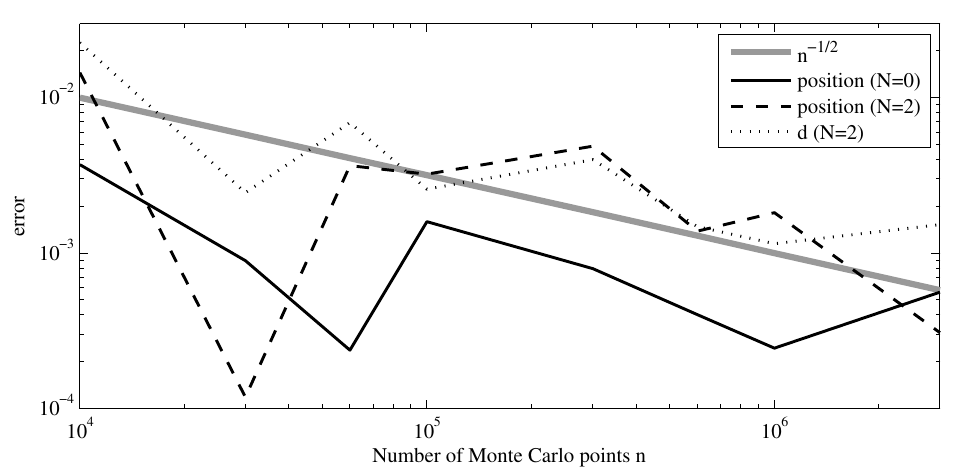}
\caption{\label{fig:hat_errors} Errors for the sampling method from~\S\ref{sec:MCMC} in dependence of the used Monte Carlo points
for the position in the case $N=0$, the position in the case $N=2$, and the observable $d$ in the case $N=2$.
The results are averaged over ten independent runs.
}
\end{figure}

Figure~\ref{fig:hat_errors} illustrates that, as expected, the asymptotic sampling error of order $O(n^{-1/2})$ for Markov chain Monte Carlo methods is
also observed for our algorithm, although the probability densities are only approximately evaluated via quadrature. 
We note that it is necessary to use a sufficiently accurate quadrature in order to observe decent convergence results.

Our experiments confirm that the Markov chain method from \S\ref{sec:MCMC} is applicable for 
the approximative sampling of Wigner functions 
in the semiclassical regime. The method could prove
particularly useful in higher dimensions, where Wigner functions typically cannot be computed.

\section*{Acknowledgements}

It is a pleasure to thank Caroline Lasser for many valuable
remarks and suggestions that helped
to improve our manuscript.
The author gratefully acknowledges support by the 
German Research Foundation (DFG), 
Collaborative Research Center SFB-TRR 109.

\begin{appendix}

\section{Proof of Proposition~\ref{lem:1d_gauss_lagu}}\label{app:proof_lem}

\begin{proof}[Proof of Proposition~\ref{lem:1d_gauss_lagu}]
We prove the assertion by induction. Since $L_0 \equiv1$, the base case $N=0$ 
is clear and we
assume that the assertion is true for some $N\in \N$. We compute
\begin{align*}
\nabla \W_{g_0}(z) &= -\tfrac2\eps z \W_{g_0}(z) & 
\nabla L_n(\r(z)) &= \tfrac4\eps z L_n'(\r(z))    \\
\Delta \W_{g_0}(z) &= \tfrac4\eps (\tfrac{\r(z)}2 -1) \W_{g_0}(z) & 
\Delta L_n(\r(z)) &= \tfrac4\eps \(  2L_n'(\r(z))  + 2\r(z) L''_n(\r(z))  \),
\end{align*}
and from now on write $ \r(z)=\r$ for simplicity. One has
\begin{align}
\tfrac{\eps}2 &\nn \Delta  (\W_{g_0}(z)L_n(\r)) = \W_{g_0}(z) \lk L_n(\r)(\r-2) + 4L_n'(\r) + 4\r L_n''(\r)  - 4 \r L_n'(\r)\rk\\
&\label{eq:laguerre_gauss_laps}= -  \W_{g_0}(z) \lk L_n(\r)(2-\r) -4 L_n'(\r) - 4\r L_n''(\r)  +  4\r L_n'(\r)\rk
\end{align}
and, hence, the polynomial factor in~\eqref{eq:laguerre_gauss_laps} can be rewritten as
\begin{align}
L_n(\r)(2-\r) &\nn - 4\r L_n''(\r)  +  4(\r-1) L_n'(\r)= L_n(\r)(2-\r+4n)\\
&\label{eq:Laguerre_reformu}=L_n(\r)(1+2n) + (n+1)L_{n+1}(\r) + nL_{n-1}(\r).
\end{align}
For verifying~\eqref{eq:Laguerre_reformu} one combines Laguerre's differential equation
\[
x L_n''(x) = (x-1)L_n'(x) - nL_n(x),
\]
and the three-term recurrence relation
\[
(n+1)L_{n+1}(x) =(2n +1 -x)L_n(x) - nL_{n-1}(x) ,\quad n \in \N,
\]
where $L_0\equiv 1$ and $L_{-1}\equiv 0$. Consequently, by  
the induction hypothesis and~\eqref{eq:Laguerre_reformu},
\begin{align}
&\nn\( -\tfrac{\eps}2 \Delta \)^{N+1}   \W_{g_0}(z)  =  N!
\sum_{n=0}^N {N \choose n}  \( -\tfrac{\eps}2 \Delta \)( \W_{g_0}(z)   L_n(\r) ) \\
&\label{eq:sum_laplace_induct}= N!  \W_{g_0}(z) 
\sum_{n=0}^N {N \choose n}  \lk L_n(\r)(1+2n) + (n+1)L_{n+1}(\r) + nL_{n-1}(\r) \rk
\end{align}
and we have to count the prefactors of the polynomials $L_n(\r)$ for
$n=0,\hdots,N+1$  in the sum.
For $L_{N+1}(\r)$ we have the prefactor
\[
(N+1) L_{N+1}(\r) = (N+1) {N+1 \choose N+1} L_{N+1}(\r)
\]
from the $N$th summand in~\eqref{eq:sum_laplace_induct}. For $L_N(\r)$ we get contributions
from the $N$th and the $(N-1)$th summand, and observe
\[
\( (1+2N) + N {N \choose N-1}  \) L_N(\r)  =   (N+1) {N+1 \choose N} L_N(\r).
\]
For all $1\leq n \leq N-1$ we get contributions from the $n$th, the $(n+1)$th, and the $(n-1)$th summand.
Combining them yields
\begin{align*}
&\({N \choose n-1}n  +{N \choose n}(1+2n) + {N \choose n+1}(n+1)\) L_n(\r) \\
&=  \(  n{N+1 \choose n} + (n+1) {N+1 \choose n+1}  \) L_n(\r) =   (N+1) {N+1 \choose n} L_n(\r),
\end{align*}
and for $L_0$ we again have the prefactor $(N+1)$. Finally, rewriting~\eqref{eq:sum_laplace_induct} 
as a sum over $L_n(\r)$ with $n=0,\hdots,N+1$ completes the proof.
\end{proof}

\section{Binomial identities}\label{sec:proof_binom}

We summarize some binomial identities we repeatedly employ in 
our proofs.
By applying Pascal's identity multiple times one directly obtains the formula
\begin{align}
\sum_{j=0}^N {k+j \choose j} &= {k+N+1 \choose N}.
\end{align}
Furthermore, for all $N,m,k\in N$ one has
\begin{equation}\label{eq:binom_sum_prod1}
\sum_{j=0}^{N}{N+m-j \choose m} {k+j \choose j} = {N+m+k+1 \choose N}.
\end{equation}
For the proof of~\eqref{eq:binom_sum_prod1} we use generating functions, and set
\[
a_N = {N+m \choose N}, \quad b_N = {N+k \choose N}, \quad c_N:= \sum_{j=0}^{N}{N+m-j \choose m} {k+j \choose j}
\]
such that
\begin{equation}\label{eq:power_series_binom}
\sum_{j\geq 0}a_j x^j = (1-x)^{-(m+1)}, \quad \sum_{j\geq 0}b_j x^j = (1-x)^{-(k+1)}
\end{equation}
for all $|x|<1$. Then, $c_N$ is the $N$th term in the Cauchy product of $(a_j)_{j\geq0}$
and $(b_j)_{j\geq 0}$, and hence
\[
\sum_{j\geq 0}c_j x^j = (1-x)^{-(m+k+2)}.
\]
Comparing the coefficients with the power series~\eqref{eq:power_series_binom} implies the assertion.

\end{appendix}

\providecommand{\noopsort}[1]{}\providecommand{\singleletter}[1]{#1}%

\end{document}